\newcommand*{\CTIV}{Vovk:arXiv0904}
  
  \newcommand*{\CTVI}{Vovk:arXiv1108}
  
  \newcommand*{\CTXI}{Vovk:arXiv1604}

  \newcommand{\PerkowskiPromelIntegrals}{Perkowski/Promel:2013}
  \newcommand{\PerkowskiPromelLocal}{Perkowski/Promel:2014}

\newcommand{\zzrelax}[1]{\relax}

\newlength{\IndentI}
\newlength{\IndentII}
\newlength{\IndentIII}
\newlength{\WidthI}
\newlength{\WidthII}
\newlength{\WidthIII}
  \documentclass[10pt]{article}
  \usepackage{amsmath,amsthm,amsfonts,amssymb,latexsym,mathtools,hyperref}
  \newcommand{\Extra}[1]{}

\allowdisplaybreaks[4]

\newcommand*{\st}{\mathrel{|}}

\newcommand*{\dd}{\,\mathrm{d}}

\newcommand*{\K}{\mathcal{K}}

\newcommand*{\BBB}{\mathcal{B}}
\newcommand*{\FFF}{\mathcal{F}}

\DeclareMathOperator{\III}{\boldsymbol{1}}  
\DeclareMathOperator{\sign}{sign}           
\DeclareMathOperator{\osc}{osc}             
\DeclareMathOperator{\tame}{tame}           

  \newcommand*{\bbbp}{\mathbb{P}}      

\DeclareMathOperator{\UpProb}{\overline{\bbbp}}    

\newcommand{\lomega}{\omega_*}   
\newcommand{\uomega}{\omega^*}   
\newcommand{\lphi}{\phi_*}       
\newcommand{\uphi}{\phi^*}       

  \newcommand*{\bbbr}{\mathbb{R}}    

  \theoremstyle{plain}
  \newtheorem{theorem}{Theorem} 

  \newtheorem{lemma}{Lemma}
  \theoremstyle{definition}
  \newtheorem{remark}{Remark}
  \newtheorem{example}{Example}

  \title{Purely pathwise probability-free It\^o integral\thanks{The version of this paper
    at \url{http://probabilityandfinance.com} (Working Paper 42)
    is updated more often.}}
  \author{Vladimir Vovk\\
  \texttt{v.vovk{\rm@}rhul.ac.uk}}

\begin{document}
  \maketitle

  \begin{abstract}
This paper gives several simple constructions of the pathwise It\^o integral $\int_0^t\phi\dd\omega$
for an integrand $\phi$ and a price path $\omega$ as integrator,
with $\phi$ and $\omega$ satisfying various topological and analytical conditions.
The definitions are purely pathwise in that neither $\phi$ nor $\omega$ are assumed to be paths of stochastic processes,
and the It\^o integral exists almost surely in a non-probabilistic financial sense.
For example, one of the results shows the existence of $\int_0^t\phi\dd\omega$ for a c\`adl\`ag integrand $\phi$
and a c\`adl\`ag integrator $\omega$ with jumps bounded in a predictable manner.
  \end{abstract}

\section{Introduction}

The structure of this paper is as follows.
To set the scene, Section~\ref{sec:literature} briefly describes papers and results that I am aware of
related to the area of probability-free pathwise It\^o integration.
In Section~\ref{sec:continuous} we define the meaning of the phrase ``some property holds almost surely''
in a probability-free manner and prove the almost sure existence of the pathwise stochastic integral $\int_0^t\phi\dd\omega$
assuming (apart from the possibility of trading in $\omega$)
that $\phi$ and $\omega$ are continuous (Theorem~\ref{thm:continuous});
the definition is ``purely pathwise'' in that neither $\omega$ nor $\phi$ are assumed to be paths of processes,
and they can be chosen separately.
Other papers (e.g., \cite{\CTIV,\PerkowskiPromelIntegrals,Lochowski:2015}) use expressions such as ``for typical outcomes''
instead of ``almost surely'' to avoid confusion with the probabilistic notion of ``almost surely'',
but there is no danger of confusion in this paper as we will never need the probabilistic notion.
Theorem~\ref{thm:continuous} is proved in the following section, Section~\ref{sec:proof};
the proof relies on a primitive ``self-normalized game-theoretic supermartingale''
introduced in Appendix~\ref{sec:A}
and a game-theoretic version of a classical martingale
introduced in Appendix~\ref{sec:B}.
The proof can also be extracted from \cite{\PerkowskiPromelIntegrals}
(which, however, does not state Theorem~\ref{thm:continuous} explicitly).
Section~\ref{sec:variation} shows that continuous price paths possess quadratic variation almost surely;
in principle, this is a known result (\cite{\CTIV}, Theorem~5.1(a)),
but we prove it in a slightly different setting (the one required for our Theorem~\ref{thm:continuous}).
Once we have the quadratic variation, we can state a simple version of It\^o's formula (Theorem~\ref{thm:Ito})
and show the coincidence of our integral with F\"ollmer's \cite{Follmer:1981} in Section~\ref{sec:Ito-lemma}.
Section~\ref{sec:cadlag} gives a definition of the It\^o integral $\int_0^t\phi\dd\omega$
in the case of c\`adl\`ag $\phi$ and $\omega$.
Theorem~\ref{thm:cadlag}, asserting the existence of It\^o integral in this case,
is proved similarly to Theorem~\ref{thm:continuous}.
The reader will notice that the setting of the former theorem is more complicated,
and so we cannot say that it contains the latter as a special case.
Section~\ref{sec:non-cadlag} makes a first step in defining purely pathwise It\^o integrals $\int_0^t\phi\dd\omega$
for non-c\`adl\`ag $\phi$ assuming, for simplicity, that $\omega$ is continuous.
Finally, Section~\ref{sec:conclusion} concludes by listing some directions of further research.

\section{Related literature}
\label{sec:literature}

The first paper to give a probability-free definition of It\^o integral was F\"ollmer's \cite{Follmer:1981},
who defined the integral $\int_0^t\phi\dd\omega$ in the case where $\phi$ is obtained
by composing a regular function $f$ (namely, $f=F'$ for a $C^2$ function $F$) with $\omega$
(for simplicity, let us assume that $\omega$ is continuous in this introductory section).
F\"ollmer's definition is pathwise in $\omega$ but not purely pathwise,
as $\phi$ is a function of $\omega$.
Cont and Fourni\'e \cite{Cont/Fournie:2010} extend F\"ollmer's results
by replacing the composition of $f$ and $\omega$ by applying a non-anticipative functional $f$
(also of the form $F'$ where $F$ is a non-anticipative functional of a class denoted $\mathbb{C}^{1,2}$
and the prime stands for ``vertical derivative'').
Cont and Fourni\'e's definition is not quite pathwise in $\omega$,
but this is repaired by Ananova and Cont in \cite{Ananova/Cont:2016}
(for the price of additional restrictions on the non-anticipative functional $F$).
Other papers (such as Perkowski and Pr\"omel \cite{\PerkowskiPromelLocal} and Davis et al.\ \cite{Davis/etal:arXiv1508})
extend F\"ollmer's results by relaxing the regularity assumptions about $f$,
which requires inclusion of local time.
All these papers assume that $\omega$ possesses quadratic variation (defined in a pathwise manner),
and this assumption is satisfied when $\omega$ is a typical price path
(see, e.g., \cite{Riga:2016};
the existence of quadratic variation for such $\omega$ was established in, e.g., \cite{\CTIV} and \cite{\CTVI};
precise definitions will be given below).
The existence of local times for typical continuous price paths follows from the main result of \cite{\CTIV}
(as explained in \cite{\PerkowskiPromelLocal}, p.~13)
and was explicitly demonstrated, together with its several nice properties,
in \cite{\PerkowskiPromelLocal} (Theorem~3.5).

The more recent paper \cite{Nutz:2012} is not completely probability-free.
Besides, it depends on additional axioms of set theory (adding the continuum hypothesis is sufficient),
and as the author points out, his ``\,{`construction'} of the stochastic integral is not `constructive' in the proper sense;
it merely yields an existence result''.
This paper's construction is explicit.

Another paper on this topic is \cite{\CTVI}, but the construction used in it is F\"ollmer's,
and the only novelty in \cite{\CTVI} is that it shows the existence of quadratic variation for typical c\`adl\`ag price paths
(under a condition bounding jumps).

To clarify the relation between the usual notion of ``pathwise'' and what we call ``purely pathwise'',
let us consider two examples in which pathwise definitions are in fact purely pathwise but very restrictive.

\begin{example}[Glenn Shafer]
  Consider the F\"ollmer-type definition of the It\^o integral $\int_0^t f(\omega(s),s)\dd\omega(s)$
  for a time-dependent function $f$
  (\cite{Sondermann:2006}, Corollary~2.3.6; this definition is implicit in \cite{Follmer:1981}).
  If $f$ does not depend on its first argument, $f(\cdot,s)=\phi(s)$,
  we obtain a purely pathwise definition of $\int_0^t\phi\dd\omega$.
  The problem is that the function $f$ has to be very regular (of class $C^{2,1}$),
  and so this construction works only for very regular $\phi$ (such as $C^1$).
\end{example}

\begin{example}\label{ex:second}
  The second example is provided by F\"ollmer's definition of the It\^o integral
  $\int_0^t\nabla F(X(s))\dd X(s)$
  for a function $X:[0,\infty)\to\bbbr^d$ having pathwise quadratic variation
  (as defined by F\"ollmer);
  this definition is given in, e.g., \cite{Follmer:1981}, pp.~147--148, and \cite{Sondermann:2006}, Theorem~2.3.4.
  Let us take $d=2$ and denote the components of $X$ as $\phi$ and $\omega$:
  $X(t)=(\phi(t),\omega(t))$ for all $t\in[0,\infty)$.
  For the existence of pathwise quadratic variation,
  it suffices to assume that $\phi$ and $\omega$ are the price paths of different securities
  in an idealized financial market
  (see, e.g., \cite{\CTVI}, Section~5).
  Taking $F(\phi,\omega):=\phi\omega$,
  we obtain the definition of the sum
  of purely pathwise It\^o integrals $\int_0^t\phi\dd\omega$ and $\int_0^t\omega\dd\phi$.
  In this special case the integrand is no longer a function of the integrator,
  but even if we ignore the fact that $\int_0^t\phi\dd\omega$ and $\int_0^t\omega\dd\phi$
  are still not defined separately,
  the fact that $\phi$ and $\omega$ are co-traded in the same market
  introduces a lot of logical dependence between them;
  e.g., in the case where $\phi(t)=\omega(t-\epsilon)$ for some $\epsilon>0$ and for all $t\ge\epsilon$
  we would expect the integral $\int_0^t\phi\dd\omega$ to be well-defined
  but a market in which such $\phi$ and $\omega$ are traded becomes a money machine
  (unless $\phi$ and $\omega$ are degenerate, such as constant).
  Even if $\phi$ is not a price path of a traded security,
  the existence of quadratic variation is a strong and unnecessary assumption.
  This paper completely decouples $\phi$ and $\omega$ (at least in the c\`adl\`ag case),
  and $\phi$ is never assumed to be a price path.
\end{example}

This paper is inspired by Rafa\l{} \L{}ochowski's recent paper \cite{Lochowski:2015},
which introduces the It\^o integral $\int_0^t\phi\dd\omega$ for a wide class of trading strategies $\phi$ as integrands
in a probability-free setting similar to that \cite{\CTIV} and \cite{\PerkowskiPromelIntegrals};
the main advance of \cite{Lochowski:2015} as compared with \cite{\PerkowskiPromelIntegrals}
is its treatment of c\`adl\`ag price processes.
The main observation leading to this paper is that $\int_0^t\phi\dd\omega$
can be defined without assuming that $\phi$ is the realized path of a given strategy.

Papers that give purely pathwise definitions of It\^o integral
include \cite{Bichteler:1981} (Theorem~7.14) and \cite{Karandikar:1995},
but the existence results in those papers are not probability-free.

Finally, on the face of it, the paper \cite{\PerkowskiPromelIntegrals} by Perkowski and Pr\"omel
does not give a purely pathwise definition
(namely, they assume the integrand to be a process rather than a path).
Perkowski and Pr\"omel consider two approaches to defining It\^o integral.
A disadvantage of their second approach is that it
``restricts the set of integrands to those {`locally looking like'}\,'' $\omega$
(\cite{\PerkowskiPromelIntegrals}, the beginning  of Section~4).
Their first approach (culminating in their Theorem~3.5) constructs $\int_0^t\phi\dd\omega$
in the case where $\phi$ is a path of a process on the sample space of continuous paths in $\bbbr^d$,
making $\phi$ a non-anticipative function of $\omega$.
It can, however, be applied to $\omega$ consisting of two components
that can be used as the integrand and the integrator
(as in Example~\ref{ex:second} above) and, crucially,
the proof of their Theorem~3.5 (see also Corollary~3.6) shows \cite{Perkowski/Promel:2016-private}
that trading in the integrand is not needed;
therefore, it also proves our Theorem~\ref{thm:continuous}.

\section{Definition of It\^o integral in the continuous case}
\label{sec:continuous}

In our terminology and definitions we will follow mainly Section~2
of the technical report \cite{\CTIV}.
We consider a game between Reality (a financial market) and Sceptic (a trader) in continuous time:
the time interval is $[0,\infty)$.
First Sceptic chooses his trading strategy (to be defined momentarily)
and then Reality chooses continuous functions $\omega$ and $\phi$ mapping $[0,\infty)$ to $\mathbb{R}$;
$\omega$ is interpreted as the price path of a financial security (not required to be nonnegative),
and $\phi$ is simply the function that we wish to integrate by $\omega$.
To formalize this picture we will be using Galmarino-style definitions,
which are more intuitive than the standard ones
(used in the journal version of \cite{\CTIV});
see, e.g., \cite{Courrege/Priouret:1965}.

Let
\begin{equation}\label{eq:sample-continuous}
  \Omega:=C[0,\infty)^2
\end{equation}
be the set of all possible pairs $(\omega,\phi)$;
it is our \emph{sample space}.
We equip $\Omega$ with the $\sigma$-algebra $\FFF$
generated by the functions $(\omega,\phi)\in\Omega\mapsto(\omega(t),\phi(t))$, $t\in[0,\infty)$
(i.e., the smallest $\sigma$-algebra making them measurable).
We often consider subsets of $\Omega$ and functions on $\Omega$
that are measurable with respect to $\FFF$.
As shown in \cite{\CTXI}, the requirement of measurability is essential:
without it, it becomes too easy to become infinitely rich infinitely quick.

As usual, an \emph{event} is an $\FFF$-measurable set in $\Omega$,
a \emph{random variable} is an $\FFF$-measurable function of the type $\Omega\to\bbbr$,
and an \emph{extended random variable} is an $\FFF$-measurable function of the type $\Omega\to[-\infty,\infty]$.
Each $o=(\omega,\phi)\in\Omega$ is identified with the function $o:[0,\infty)\to\bbbr^2$
defined by
\begin{equation*}
  o(t)
  :=
  (\omega(t),\phi(t)),
  \quad
  t\in[0,\infty).
\end{equation*}
A \emph{stopping time} is an extended random variable $\tau:\Omega\to[0,\infty]$ such that,
for all $o$ and $o'$ in $\Omega$,
\begin{equation*}
  \left(
    o|_{[0,\tau(o)]}
    =
    o'|_{[0,\tau(o)]}
  \right)
  \Longrightarrow
  \tau(o)=\tau(o'),
\end{equation*}
where $f|_A$ stands for the restriction of $f$ to the intersection of $A$ and $f$'s domain.
A random variable $X$ is said to be \emph{$\tau$-measurable},
where $\tau$ is a stopping time,
if,
for all $o$ and $o'$ in $\Omega$,
\begin{equation*}
  \left(
    o|_{[0,\tau(o)]}
    =
    o'|_{[0,\tau(o)]}
  \right)
  \Longrightarrow
  X(o)=X(o').
\end{equation*}
As customary in probability theory,
we will often omit explicit mention of $o\in\Omega$ when it is clear from the context.

A \emph{simple trading strategy} $G$ is defined to be a pair $((\tau_1,\tau_2,\ldots),(h_1,h_2,\ldots))$,
where:
\begin{itemize}
\item
  $\tau_1\le\tau_2\le\cdots$ is a nondecreasing sequence of stopping times
  such that, for each $o\in\Omega$,
  $\lim_{n\to\infty}\tau_n(o)=\infty$;
\item
  for each $n=1,2,\ldots$, $h_n$ is a bounded $\tau_{n}$-measurable function.
\end{itemize}
The \emph{simple capital process} corresponding to a simple trading strategy $G$
and an \emph{initial capital} $c\in\bbbr$ is defined, for $o=(\omega,\phi)$, by
\begin{equation*} 
  \K^{G,c}_t(o)
  :=
  c
  +
  \sum_{n=1}^{\infty}
  h_n(o)
  \bigl(
    \omega(\tau_{n+1}\wedge t)-\omega(\tau_n\wedge t)
  \bigr),
  \quad
  t\in[0,\infty),
\end{equation*}
where the zero terms in the sum are ignored
(which makes the sum finite for each $t$).
The value $h_n(o)$ is Sceptic's \emph{bet} at time $\tau_n=\tau_n(o)$,
and $\K^{G,c}_t(o)$ is Sceptic's \emph{capital} at time $t$.
The intuition behind this definition is that Sceptic is allowed to bet only on $\omega$,
but the current and past values of both $\omega$ and $\phi$ can be used for choosing the bets.

A \emph{nonnegative capital process} is any function $\mathfrak{S}$
that can be represented in the form
\begin{equation}\label{eq:positive-capital}
  \mathfrak{S}_t
  :=
  \sum_{n=1}^{\infty}
  \K^{G_n,c_n}_t,
\end{equation}
where the simple capital processes $\K^{G_n,c_n}$
are required to be nonnegative
(i.e., $\K^{G_n,c_n}_t(o)\ge0$ for all $t$ and $o\in\Omega$),
and the nonnegative series $\sum_{n=1}^{\infty}c_n$ is required to converge.
The sum \eqref{eq:positive-capital} can take value $\infty$.
Since $\K^{G_n,c_n}_0(o)=c_n$ does not depend on $o$,
$\mathfrak{S}_0(o)$ does not depend on $o$ either
and will sometimes be abbreviated to $\mathfrak{S}_0$.

The \emph{outer measure} of a set $E\subseteq\Omega$ (not necessarily $E\in\FFF$)
is defined as
\begin{equation}\label{eq:upper-probability}
  \UpProb(E)
  :=
  \inf
  \bigl\{
    \mathfrak{S}_0
    \bigm|
    \forall o\in\Omega:
    \liminf_{t\to\infty}
    \mathfrak{S}_t(o)
    \ge
    \III_E(o)
  \bigr\},
\end{equation}
where $\mathfrak{S}$ ranges over the nonnegative capital processes
and $\III_E$ stands for the indicator function of $E$.
It is clear that $\UpProb(E)$ will not change if we replace the $\liminf$ in~\eqref{eq:upper-probability} by $\limsup$.
The set $E$ is \emph{null} if $\UpProb(E)=0$.
This condition is equivalent to the existence of a nonnegative capital process $\mathfrak{S}$
such that $\mathfrak{S}_0=1$ and, on the event $E$, $\lim_{t\to\infty}\mathfrak{S}_t=\infty$
(see, e.g., \cite{\CTIV}, Section~2).
The equivalence will still hold if we replace the $\lim_{t\to\infty}$ by $\limsup_{t\to\infty}$.
We will say that such $\mathfrak{S}$ \emph{witnesses} that $E$ is null.
A property of $o\in\Omega$ will be said to hold \emph{almost surely} (a.s.)\
if the set of $o$ where it fails is null.

\begin{remark}
  The definition \eqref{eq:upper-probability} is less popular
  than its modification proposed in \cite{\PerkowskiPromelIntegrals}
  (the latter has been also used in, e.g., \cite{\PerkowskiPromelLocal}, \cite{Beiglbock/etal:2015},
  \cite{Lochowski:2015LMJ-arXiv}, and \cite{Lochowski:2015}).
  Our rationale for the choice of the original definition \eqref{eq:upper-probability}
  is that it is more conservative and, therefore, makes our results stronger.
  Its financial interpretation is that $E$ is null
  if Sceptic can become infinitely rich
  splitting an initial capital of only one monetary unit
  into a countable number of accounts
  and running a simple trading strategy on each account
  making sure that no account ever goes into debt.
\end{remark}

Now we have all we need to define the It\^o integral $\int_0^t\phi\dd\omega$.
First we define a sequence of stopping times $T^n_k$, $k=0,1,2,\ldots$, inductively
by $T^n_{0}(o):=0$, where $o=(\omega,\phi)$, and
\begin{equation}\label{eq:T}
  T^n_k(o)
  :=
  \inf
  \Bigl\{
    t>T^n_{k-1}(o)
    \st
    \left|\omega(t)-\omega(T^n_{k-1})\right| \vee \left|\phi(t)-\phi(T^n_{k-1})\right| = 2^{-n}
  \Bigr\}
\end{equation}
for $k=1,2,\ldots$ (as usual, $\inf\emptyset:=\infty$);
we do this for each $n=1,2,\ldots$\,.
We let $T^n(o)$ stand for the \emph{$n$th partition}, i.e., the set
$$
  T^n(o)
  :=
  \left\{
    T^n_k(o)
    \st
    k=0,1,\ldots
  \right\}.
$$
Notice that the nestedness of the partitions,
$T^1\subseteq T^2\subseteq\cdots$,
is neither required nor implied by our definition.

\begin{remark}
  The definition of the sequence \eqref{eq:T} is different from the one in \cite{\CTIV}, Section~5,
  in that it uses not only the values of $\omega$ but also those of $\phi$.
  In this respect it is reminiscent of the definitions in \cite{Bichteler:1981} (Theorem~7.14) and \cite{Karandikar:1995},
  where similar sequences of stopping times depend only on the values of $\phi$.
\end{remark}

For all $t\in[0,\infty)$, $\phi\in C[0,\infty)$, and $\omega\in C[0,\infty)$, define
\begin{equation}\label{eq:integral}
  (\phi\cdot\omega)^n_t
  :=
  \sum_{k=1}^{\infty}
  \phi(T^n_{k-1}\wedge t)
  \Bigl(
    \omega(T^n_{k}\wedge t)
    -
    \omega(T^n_{k-1}\wedge t)
  \Bigr),
  \quad
  n=1,2,\ldots\,.
\end{equation}

\begin{theorem}\label{thm:continuous}
  For each $t>0$,
  $(\phi\cdot\omega)^n_s$
  converges uniformly over $s\in[0,t]$ almost surely
  as $n\to\infty$.
\end{theorem}

The limit whose existence is asserted in Theorem~\ref{thm:continuous} will be denoted $(\phi\cdot\omega)_s$
or $\int_0^s\phi\dd\omega$ and called the It\^o integral of $\phi$ by $\omega$.
Since the convergence is uniform over $[0,t]$ for each $t$,
$(\phi\cdot\omega)_s$ is a continuous function of $s\in[0,\infty)$,
almost surely.

\section{Proof of Theorem~\protect\ref{thm:continuous}}
\label{sec:proof}

Let us first check the following basic property of the stopping times $T^n_k$
(which will allow us to use these stopping times as components of simple trading strategies).

\begin{lemma}\label{lem:basic}
  For each $n$, $T^n_k\to\infty$ as $k\to\infty$.
\end{lemma}

\begin{proof}
  Let us fix $n$ and $t$ and show that $T^n_k>t$ for some $k$.
  Each $s\in[0,t]$ has a neighbourhood in which $\omega$ and $\phi$ change by less than $2^{-n}$.
  By the compactness of the interval $[0,t]$ we can choose a finite cover of this interval
  consisting of such neighbourhoods,
  and each such neighbourhood contains at most one $T^n_k$.
\end{proof}

By the definition of ``almost surely'' (and the remarks preceding it),
any nonnegative capital process $\K$ with $\K_0<\infty$ satisfies $\sup_{t\in[0,\infty)}\K_t<\infty$ almost surely.
The following lemma generalizes this to sequences of nonnegative capital processes.

\begin{lemma}\label{lem:O}
  For any sequence $\K^n$, $n=1,2,\ldots$, of nonnegative capital processes satisfying $\K^n_0\le1$,
  we have $\sup_{t\in[0,\infty)}\K^n_t=O(n^2)$ as $n\to\infty$ a.s.
\end{lemma}

\begin{proof}
  Since $\sum_n n^{-2}\K^n$ is a nonnegative capital process with a finite initial value,
  we have
  $$
    \sup_{t}\sum_n n^{-2}\K^n_t<\infty \quad \text{a.s.},
  $$
  which implies
  $$
    \sup_{t}\sup_n n^{-2}\K^n_t<\infty \quad \text{a.s.},
  $$
  which in turn implies
  $$
    \sup_n n^{-2}\sup_{t}\K^n_t<\infty \quad \text{a.s.},
  $$
  which can be rewritten as $\sup_{t}\K^n_t=O(n^2)$ a.s.
\end{proof}

The value of $t$ will be fixed throughout the rest of this section.
It suffices to prove that the sequence of functions
$(\phi\cdot\omega)^n_s$ on the interval $s\in[0,t]$ is Cauchy (in the uniform metric) almost surely.

Let us arrange the stopping times $T^n_0,T^n_1,T^n_2,\ldots$
and $T^{n-1}_0,T^{n-1}_1,T^{n-1}_2,\ldots$ into one strictly increasing sequence
(removing duplicates if needed) $a_k$, $k=0,1,\ldots$:
$0=a_0<a_1<a_2<\cdots$,
each $a_k$ is equal to one of the $T^n_k$ or one of the $T^{n-1}_k$,
each $T^n_k$ is among the $a_k$,
and each $T^{n-1}_k$ is among the $a_k$.
Let us apply the strategy leading to the supermartingale~\eqref{eq:super-1}
(eventually we will be interested in~\eqref{eq:super-2}) to
\begin{equation}\label{eq:x}
  \begin{aligned}
    x_k
    &:=
    b_n
    \left(
      \left(
        (\phi\cdot\omega)^n_{a_k}
        -
        (\phi\cdot\omega)^n_{a_{k-1}}
      \right)
      -
      \left(
        (\phi\cdot\omega)^{n-1}_{a_k}
        -
        (\phi\cdot\omega)^{n-1}_{a_{k-1}}
      \right)
    \right)\\
    &=
    b_n
    \Bigl(
      \phi(a'_{k-1})
      \left(
        \omega(a_k)
        -
        \omega(a_{k-1})
      \right)
      -
      \phi(a''_{k-1})
      \left(
        \omega(a_k)
        -
        \omega(a_{k-1})
      \right)
    \Bigr)\\
    &=
    b_n
    \left(
      \phi(a'_{k-1})
      -
      \phi(a''_{k-1})
    \right)
    \left(
      \omega(a_k)
      -
      \omega(a_{k-1})
    \right),
  \end{aligned}
\end{equation}
where
$b_n:=n^{2}$
(the rationale for this choice will become clear later),
$a'_{k-1}:=T^n_{k'}$ with $k'$ being the largest integer such that $T^{n}_{k'}\le a_{k-1}$,
and $a''_{k-1}:=T^{n-1}_{k''}$ with $k''$ being the largest integer such that $T^{n-1}_{k''}\le a_{k-1}$.
(Notice that either $a'_{k-1}=a_{k-1}$ or $a''_{k-1}=a_{k-1}$.)
Informally, we consider the simple capital process $\K^n$ with starting capital 1
corresponding to betting $\K^n_{a_k}$ on $x_k$ at each time $a_k$, $k=0,1,\ldots$\,.
Formally, the bet (on $\omega$) at time $a_k$ is
$
  \K^n_{a_k}
  b_n
  \left(
    \phi(a'_{k})
    -
    \phi(a''_{k})
  \right)
$.

We often do not reflect $n$ in our notation (such as $a_k$ and $x_k$),
but this should not lead to ambiguities.

The condition of Lemma~\ref{lem:super} is satisfied as
\begin{equation}\label{eq:bound-1}
  \left|x_k\right|
  \le
  b_n 2^{-n+1} 2^{-n}
  \le
  0.5,
\end{equation}
where the last inequality
(ensuring that~\eqref{eq:super-1} and~\eqref{eq:super-2} are really supermartingales)
is true for all $n\ge1$.
By Lemma~\ref{lem:super}, we will have
\begin{equation*}
  \K^n_{a_K}
  \ge
  \prod_{k=1}^K
  \exp(x_k-x_k^2),
  \quad
  K=0,1,\ldots\,.
\end{equation*}
The proof of Lemma~\ref{lem:super} shows that, in addition, we will also have
\begin{equation*}
  \K^n_s
  \ge
  \K^n_{a_{k-1}}
  \exp(x_{k,s}-x_{k,s}^2),
  \quad
  k=1,2,\ldots,
  \quad
  s\in[a_{k-1},a_{k}],
\end{equation*}
where
\begin{equation}\label{eq:x-s}
  \begin{aligned}
    x_{k,s}
    &:=
    b_n
    \left(
      \left(
        (\phi\cdot\omega)^n_s
        -
        (\phi\cdot\omega)^n_{a_{k-1}}
      \right)
      -
      \left(
        (\phi\cdot\omega)^{n-1}_s
        -
        (\phi\cdot\omega)^{n-1}_{a_{k-1}}
      \right)
    \right)\\
    &=
    b_n
    \left(
      \phi(a'_{k-1})
      -
      \phi(a''_{k-1})
    \right)
    \left(
      \omega(s)
      -
      \omega(a_{k-1})
    \right)
  \end{aligned}
\end{equation}
(cf.\ \eqref{eq:x};
notice that \eqref{eq:bound-1} remains true for $x_{k,s}$ in place of $x_k$).
This simple capital process $\K^n$ is obviously nonnegative.

To cover both \eqref{eq:x} and \eqref{eq:x-s}, we modify \eqref{eq:x-s} to
\begin{equation}\label{eq:x-ks}
  x_{k,s}
  :=
  b_n
  \left(
    \phi(a'_{k-1})
    -
    \phi(a''_{k-1})
  \right)
  \left(
    \omega(a_k \wedge s)
    -
    \omega(a_{k-1} \wedge s)
  \right).
\end{equation}
We have a nonnegative capital process $\K^n$ that starts from $1$
and whose value at time $s$ is at least
\begin{equation}\label{eq:final-1}
  \exp
  \left(
    b_n
    \left(
      (\phi\cdot\omega)^n_{s}
      -
      (\phi\cdot\omega)^{n-1}_{s}
    \right)
    -
    \sum_{k=1}^{\infty} x_{k,s}^2
  \right).
\end{equation}

Let us show that
\begin{equation}\label{eq:square}
  \sup_{s\in[0,t]}
  \sum_{k=1}^{\infty} x_{k,s}^2=o(1)
\end{equation}
as $n\to\infty$ almost surely.
It suffices to show that
\begin{equation}\label{eq:sup}
  \sup_{s\in[0,t]}
  \sum_{k=1}^{\infty}
  \left(
    n^2 2^{-n+1}
    \left(
      \omega(a_k\wedge s)-\omega(a_{k-1}\wedge s)
    \right)
  \right)^2
  =
  o(1)
\end{equation}
almost surely.
Using the trading strategy leading to the K29 martingale \eqref{eq:K29},
we obtain the simple capital process
\begin{align}
    \tilde\K^n_s
    &=
    n^{-3}
    +
    \sum_{k=1}^{\infty}
    \left(
      n^2 2^{-n+1}
      \left(
        \omega(a_k\wedge s)-\omega(a_{k-1}\wedge s)
      \right)
    \right)^2
    \notag\\
    &\qquad-
    \left(
      \sum_{k=1}^{\infty}
      n^2 2^{-n+1}
      \left(
        \omega(a_k\wedge s)-\omega(a_{k-1}\wedge s)
      \right)
    \right)^2
    \notag\\
    &=
    n^{-3}
    +
    \sum_{k=1}^{\infty}
    n^4 2^{-2n+2}
    \left(
      \omega(a_k\wedge s)-\omega(a_{k-1}\wedge s)
    \right)^2
    \label{eq:sum}\\
    &\qquad-
    n^4 2^{-2n+2}
    \left(
      \omega(s)-\omega(0)
    \right)^2.
    \notag
\end{align}
Formally, this simple capital process corresponds to the initial capital $\tilde\K^n_0=n^{-3}$
and betting $-2 n^4 2^{-2n+2} (\omega(a_k)-\omega(0))$ at time $a_k$, $k=1,2,\ldots$
(cf.~\eqref{eq:increment} on p.~\pageref{eq:increment}).
Let us make this simple capital process nonnegative
by stopping trading at the first moment $s$ when $n^4 2^{-2n+2}\left(\omega(s)-\omega(0)\right)^2$ reaches $n^{-3}$
(which will not happen before time $t$ for sufficiently large $n$);
notice that this will make $\tilde\K^n$ nonnegative
even if the addend $\sum_{k=1}^{\infty}(\cdots)^2$ in~\eqref{eq:sum} is ignored.
Since $\tilde\K^n$ is a nonnegative capital process with initial value $n^{-3}$,
applying Lemma~\ref{lem:O} to $n^3\tilde\K^n$ gives $\sup_{s\le t}\tilde\K^n_s=O(n^{-1})=o(1)$ a.s.
Therefore, the sum $\sum_{k=1}^{\infty}(\cdots)^2$ in~\eqref{eq:sum} is $o(1)$ a.s.,
which completes the proof of~\eqref{eq:square}.

In combination with \eqref{eq:square}, \eqref{eq:final-1} implies
\begin{equation*}
  \K^n_s
  \ge
  \exp
  \Bigl(
    b_n
    \left(
      (\phi\cdot\omega)^n_s
      -
      (\phi\cdot\omega)^{n-1}_s
    \right)
    -
    1
  \Bigr)
\end{equation*}
for all $s\le t$ from some $n$ on almost surely.
Applying the strategy leading to the supermartingale~\eqref{eq:super-1} to $-x_{k,s}$ in place of $x_{k,s}$
and averaging the resulting simple capital processes (as in \eqref{eq:super-2}),
we obtain a simple capital process $\bar\K^n$ satisfying
\begin{equation}\label{eq:final-2}
  \bar\K^n_s
  \ge
  \frac12
  \exp
  \Bigl(
    b_n
    \left|
      (\phi\cdot\omega)^n_s
      -
      (\phi\cdot\omega)^{n-1}_s
    \right|
    -
    1
  \Bigr)
\end{equation}
for all $s\le t$ from some $n$ on almost surely.

By the definition of $\bar\K^n$ and Lemma~\ref{lem:O}, we obtain that
\begin{equation*}
  \sup_{s\in[0,t]}
  \frac12
  \exp
  \left(
    n^{2}
    \left|
      (\phi\cdot\omega)^n_s
      -
      (\phi\cdot\omega)^{n-1}_s
    \right|
    -
    1
  \right)
  =
  O(n^2)
\end{equation*}
almost surely.
The last inequality implies
\begin{equation*}
  \sup_{s\in[0,t]}
  \left|
    (\phi\cdot\omega)^n_s
    -
    (\phi\cdot\omega)^{n-1}_s
  \right|
  =
  O
  \left(
    \frac{\log n}{n^{2}}
  \right).
\end{equation*}
Since the series $\sum_n(\log n)n^{-2}$ converges,
we have the almost sure uniform convergence of $(\phi\cdot\omega)^n_s$ over $s\in[0,t]$
as $n\to\infty$.

\section{Quadratic variation}
\label{sec:variation}

In this section we will show that the quadratic variation of $\omega$ along $T^n_k$ exists.
This was shown in, e.g., \cite{\CTIV} and \cite{\CTVI},
but for a different sequence of partitions.

Define (essentially following \cite{\CTIV}, Section~5)
\begin{equation*}
  A^n_t(o)
  :=
  \sum_{k=1}^{\infty}
  \left(
    \omega(T^n_k\wedge t)
    -
    \omega(T^n_{k-1}\wedge t)
  \right)^2,
  \quad
  n=1,2,\ldots,
\end{equation*}
for $o=(\omega,\phi)$.

\begin{lemma}\label{lem:QV}
  For each $t\ge0$, the sequence of functions $A^n:s\in[0,t]\mapsto A^n_s$ converges as $n\to\infty$ uniformly
  almost surely.
\end{lemma}

We will use the notation $A_s(o)$ for the limit (when it exists)
and will call it the \emph{quadratic variation of $\omega$ at $s$}.
We will also use the notation $A(o)$
for the \emph{quadratic variation} $s\ge0\mapsto A_s(o)$ of the price path $\omega$.

\begin{proof}[Proof of Lemma~\ref{lem:QV}]
  The proof will be modelled on that of Theorem~\ref{thm:continuous} in Section~\ref{sec:proof}
  (but will be simpler);
  we start from fixing the value of $t$.
  Let us check that the sequence $A^n|_{[0,t]}$
  is Cauchy in the uniform metric almost surely.

  Let us apply the supermartingale~\eqref{eq:super-1} to
  \begin{align*}
    x_k
    &:=
    b_n
    \left(
      \left(
        A^n_{a_k}(o)
        -
        A^n_{a_{k-1}}(o)
      \right)
      -
      \left(
        A^{n-1}_{a_k}(o)
        -
        A^{n-1}_{a_{k-1}}(o)
      \right)
    \right)\\
    &=
    b_n
    \Bigl(
      \left(
        \omega(a_k)
        -
        \omega(a'_{k-1})
      \right)^2
      -
      \left(
        \omega(a_{k-1})
        -
        \omega(a'_{k-1})
      \right)^2\\
      &\qquad-
      \left(
        \omega(a_k)
        -
        \omega(a''_{k-1})
      \right)^2
      +
      \left(
        \omega(a_{k-1})
        -
        \omega(a''_{k-1})
      \right)^2
    \Bigr)\\
    &=
    b_n
    \Bigl(
      -2 \omega(a_k) \omega(a'_{k-1})
      +
      2 \omega(a_{k-1}) \omega(a'_{k-1})\\
      &\qquad+
      2 \omega(a_k) \omega(a''_{k-1})
      -
      2 \omega(a_{k-1}) \omega(a''_{k-1})
    \Bigr)\\
    &=
    2b_n
    \left(
      \omega(a''_{k-1}) - \omega(a'_{k-1})
    \right)
    \left(
      \omega(a_k) - \omega(a_{k-1})
    \right)
  \end{align*}
  and to $-x_k$,
  where $a'_{k-1}$, $a''_{k-1}$, and $b_n$ are defined as before
  and we are interested only in $n\ge4$.
  Instead of the bound \eqref{eq:bound-1} we now have
  \begin{equation*}
    \left|x_k\right|
    \le
    2b_n 2^{-n+1} 2^{-n}
    =
    b_n 2^{-2n+2}
    \le
    0.5
  \end{equation*}
  (the last inequality depending on our assumption $n\ge4$).
  The analogue of \eqref{eq:final-2} is
  \begin{equation*}
    \bar\K^n_s
    \ge
    \frac12
    \exp
    \Bigl(
      b_n
      \left|
        A^n_s(o)
        -
        A^{n-1}_s(o)
      \right|
      -
      1
    \Bigr),
  \end{equation*}
  and so we have
  \begin{equation*}
    \sup_{s\in[0,t]}
    \frac12
    \exp
    \left(
      n^{2}
      \left|
        A^n_s
        -
        A^{n-1}_s
      \right|
      -
      1
    \right)
    =
    O(n^2)
  \end{equation*}
  almost surely.
  This implies
  \begin{equation*}
    \sup_{s\in[0,t]}
    \left|
      A^n_s
      -
      A^{n-1}_s
    \right|
    =
    O
    \left(
      \frac{\log n}{n^{2}}
    \right)
  \end{equation*}
  and thus the almost sure uniform convergence of $A^n_s$ over $s\in[0,t]$ as $n\to\infty$.
\end{proof}

\section{It\^o's formula}
\label{sec:Ito-lemma}

In this section we state a version of It\^o's formula which shows that our definition of It\^o integral
agrees with that of F\"ollmer \cite{Follmer:1981}
(when the latter is specialized to the continuous case and our sequence of partitions).

\begin{theorem}\label{thm:Ito}
  Let $F:\mathbb{R}\to\mathbb{R}$ be a function of class $C^2$.
  Then, for all $t\ge0$,
  \begin{equation}\label{eq:Ito}
    F(\omega(t))
    =
    F(\omega(0))
    +
    \int_0^t
    F'(\omega)
    \dd\omega
    +
    \frac12
    \int_0^t
    F''(\omega)
    \dd A(\omega,F'(\omega))
    \text{\quad a.s.}
  \end{equation}
\end{theorem}

The notation $F'(\omega)$ and $F''(\omega)$ in \eqref{eq:Ito} stands for compositions:
e.g., $F'(\omega)(s):=F'(\omega(s))$ for $s\ge0$.
The integral
$
  \int_0^t
  F''(\omega)
  \dd A(\omega,F'(\omega))
$
can be understood in the Lebesgue--Stiltjes sense.
The arguments ``$(\omega,F'(\omega))$'' of $A$ refer to the sequence of partitions
(\eqref{eq:T} with $\phi:=F'(\omega)$) used when defining $A$.

\begin{proof}
  By Taylor's formula,
  \begin{multline*}
    F(\omega(T^n_k)) - F(\omega(T^n_{k-1}))
    =
    F'(\omega(T^n_{k-1}))
    \Bigl(
      \omega(T^n_k) - \omega(T^n_{k-1})
    \Bigr)\\
    +
    \frac12
    F''(\xi_k)
    \Bigl(
      \omega(T^n_k) - \omega(T^n_{k-1})
    \Bigr)^2,
  \end{multline*}
  where $\xi_k\in[\omega(T^n_{k-1}),\omega(T^n_{k})]$.
  It remains to sum this equality over $k=1,\ldots,K$,
  where $K$ is the largest $k$ such that $T^n_k\le t$,
  and to pass to the limit as $n\to\infty$.
\end{proof}

Since It\^o's formula \eqref{eq:Ito} holds for F\"ollmer's \cite{Follmer:1981} integral $\int_0^t F'(\omega)\dd\omega$ as well
(see the theorem in \cite{Follmer:1981}),
F\"ollmer's integral (defined only in the context of $\int F'(\omega)\dd\omega$) coincides with ours almost surely.
This is true for the sequence of partitions \eqref{eq:T} with $\phi:=F'(\omega)$,
provided it is dense (as required in F\"ollmer's definitions,
which in this case is equivalent to ours:
cf.\ \cite{\CTVI}, Proposition~4).

\section{The case of c\`adl\`ag integrand and integrator}
\label{sec:cadlag}

In this section we allow $\omega$ and $\phi$ to be c\`adl\`ag functions,
and this requires adding further components to Reality's move,
c\`adl\`ag functions $\lomega$ and $\uomega$ that control the jumps of $\omega$ in a predictable manner.
The sample space $\Omega$ (the set of all possible moves by Reality) now becomes
\begin{equation}\label{eq:sample-cadlag}
  \Omega
  :=
  \left\{
    (\omega,\lomega,\uomega,\phi) \in D[0,\infty)^4
    \st
    \forall t\in(0,\infty): \lomega(t-)\le\omega(t)\le\uomega(t-)
  \right\},
\end{equation}
where $D[0,\infty)$ is the Skorokhod space of all c\`adl\`ag real-valued functions on $[0,\infty)$,
and $f(t-)$ stands for the left limit $\lim_{s\uparrow t}f(s)$ of $f$ at $t>0$.

The $\Omega$ of the previous section, \eqref{eq:sample-continuous},
embeds into the $\Omega$ of this section, \eqref{eq:sample-cadlag},
by setting $\lomega:=\omega$ and $\uomega:=\omega$.

\begin{remark}
  The condition on the jumps of $\omega$ given in \eqref{eq:sample-cadlag} is similar to the condition given in \cite{\CTVI},
  which assumes that $\lomega$ and $\uomega$ are functions of $\omega$
  (i.e., that there are functions $f_*$ and $f^*$ such that $\lomega(t)=f_*(\omega(t))$ and $\uomega(t)=f^*(\omega(t))$ for all $t>0$)
  and that $\omega=(\lomega+\uomega)/2$.
  It covers two important special cases:
  \begin{itemize}
  \item
    the jumps $\Delta\omega(t):=\omega(t)-\omega(t-)$ of $\omega$ are bounded by a known constant $C$ in absolute value;
    this corresponds to $\lomega:=\omega-C$ and $\uomega:=\omega+C$;
  \item
    $\omega$ is known to be nonnegative (as price paths in real-world markets often are)
    and the relative jumps $\Delta\omega(t)/\omega(t-)$ (with $0/0:=0$) are bounded above by a known constant $C$;
    this corresponds to $\lomega:=0$ and $\uomega:=(1+C)\omega$.
  \end{itemize}
\end{remark}

Each $o=(\omega,\lomega,\uomega,\phi)\in\Omega$ is identified with the function $o:[0,\infty)\to\bbbr^4$
defined by
\begin{equation*}
  o(t)
  :=
  (\omega(t),\lomega(t),\uomega(t),\phi(t)),
  \quad
  t\in[0,\infty).
\end{equation*}
The sample space $\Omega$ is equipped with the universal completion $\FFF$
of the $\sigma$-algebra generated by the functions $o\in\Omega\mapsto o(t)$, $t\in[0,\infty)$.
After this change, the definitions of events, random variables,
stopping times $\tau$, and $\tau$-measurable random variables
remain as before (but with the new $\sigma$-algebra $\FFF$).

We need universal completion in the definition of $\FFF$ to have the following lemma.

\begin{lemma}\label{lem:bona-fide}
  If $A\subseteq\bbbr$ is a closed set, its entry time by $\omega$,
  \begin{equation*}
    \tau(o)
    :=
    \min\{t\in[0,\infty):\omega(t)\in A\},
  \end{equation*}
  $o$ standing for $(\omega,\lomega,\uomega,\phi)$,
  is a stopping time.
\end{lemma}

\begin{proof}
  See, e.g., the third example in \cite{Dellacherie:1978}
  (combined with the universal measurability of analytic sets,
  Theorem III.33 in \cite{Dellacherie/Meyer:1978}).
  For completeness, however, we will spell out the simple argument.
  Since $A$ is closed, $\{\tau\le t\}$ is the projection onto $\Omega$ of the set
  $\{(s,o)\in[0,t]\times\Omega\st\omega(s)\in A\}$.
  In combination with the progressive measurability of c\`adl\`ag processes
  (such as $\mathfrak{S}_s(o):=\omega(s)$)
  this implies that,
  since $\{(s,o)\in[0,t]\times\Omega\st\omega(s)\in A\}$ is in the $\sigma$-algebra $\BBB_t\times\FFF_t$
  (where $\BBB_t$ stands for the Borel $\sigma$-algebra on $[0,t]$),
  $\{\tau\le t\}$ is analytic.
\end{proof}

\begin{remark}\label{rem:bona-fide}
  The analogues of Lemma~\ref{lem:bona-fide} also hold for $\phi$, $\lomega$, and $\uomega$ in place of $\omega$
  (as the same argument shows).
\end{remark}

The definitions of a simple trading strategy, a simple capital process,
a nonnegative capital process, and the outer measure
stay the same as in Section~\ref{sec:continuous}
apart from replacing the argument ``$o=(\omega,\phi)$'' by ``$o=(\omega,\lomega,\uomega,\phi)$'';
``almost sure'' is also defined as before.

The definition \eqref{eq:T} of $T^n_k$ is modified by replacing
the equality with an inequality:
$T^n_0(o):=0$ and
\begin{multline}\label{eq:T-cadlag}
  T^n_k(o)
  :=
  \inf
  \Bigl\{
    t>T^n_{k-1}(o)
    \st
    \left|\omega(t)-\omega(T^n_{k-1})\right| \vee \left|\phi(t)-\phi(T^n_{k-1})\right| \ge 2^{-n}
  \Bigr\},\\
  k=1,2,\ldots\,.
\end{multline}
After this change is made, the definition of $(\phi\cdot\omega)^n_s$ stays as before,
\eqref{eq:integral}.
The analogue of Lemma~\ref{lem:basic} still holds:

\begin{lemma}
  For each $n$, $T^n_k\to\infty$ as $k\to\infty$.
\end{lemma}

\begin{proof}
  The proof is analogous to the proof of Lemma~\ref{lem:basic},
  except that now we choose a neighbourhood of each $s\in[0,t]$
  in which $\omega$ changes by less than $\left|\Delta\omega(s)\right|+2^{-n}$
  and $\phi$ changes by less than $\left|\Delta\phi(s)\right|+2^{-n}$.
  In each such neighbourhood there are less than 10 values of $T^n_k$ (for a fixed $n$).
\end{proof}

The following theorem asserts the almost sure existence of It\^o integral
in our current context.

\begin{theorem}\label{thm:cadlag}
  For each $t>0$,
  $(\phi\cdot\omega)^n_s$
  converges uniformly over $s\in[0,t]$ almost surely
  as $n\to\infty$.
\end{theorem}

\begin{proof}
  Fix $t>0$ and let $E$ be the event that $(\phi\cdot\omega)^n_s$ fails to converge uniformly over $s\in[0,t]$
  as $n\to\infty$.
  Our goal is to prove that $E$ is null.
  Assume, without loss of generality, that $\omega(0)=0$
  (this can be done as \eqref{eq:integral} is invariant with respect to adding a constant to $\omega$).

  First we notice (as in the proof of Theorem~1 of \cite{\CTVI})
  that it suffices to consider the modified game in which Reality does not output $\lomega$ and $\uomega$
  but instead is constrained
  to producing $\omega\in D[0,\infty)$ satisfying $\sup_{s\in[0,\infty)}\left|\omega(s)\right|\le c$
  for a given constant $c>0$.
  Indeed, suppose that the statement of Theorem~\ref{thm:cadlag} (for the given $t$)
  holds in the modified game for any $c$,
  and let $\mathfrak{S}^c$ be a nonnegative capital process
  witnessing that the analogue of the event $E$ in the modified game is null.
  A nonnegative capital process $\mathfrak{S}$ witnessing that $E$ is null in the original game
  can be defined as
  \begin{equation}\label{eq:witness}
    \mathfrak{S}_s
    :=
    \sum_{L=1}^{\infty}
    2^{-L}
    \mathfrak{S}^{2^L}_{s\wedge\sigma_L}
  \end{equation}
  where $\sigma_L$ is the stopping time
  \begin{equation*}
    \sigma_L
    :=
    \inf
    \left\{
      s
      \st
      \uomega(s)\vee(-\lomega(s))
      \ge
      2^L
    \right\}
  \end{equation*}
  (intuitively $\sigma_L$ is the first moment when we can no longer guarantee that $\omega$
  will not jump to or above $2^L$ in absolute value straight away;
  this is a stopping time by Lemma~\ref{lem:bona-fide} and Remark~\ref{rem:bona-fide}).
  Let us check that each addend in \eqref{eq:witness} is nonnegative not only in the modified but also in the original game.
  Indeed, if $\mathfrak{S}^{2^L}_s<0$ for some $s\le\sigma_L$,
  the nonnegativity of $\mathfrak{S}^{2^L}$ in the modified game (with $c=2^L$) implies that,
  for some $s'\in[0,s]$,
  $\left|\omega(s')\right|>2^L$.
  By~\eqref{eq:sample-cadlag},
  the last inequality implies $\uomega(s'-)>2^L$ or $\lomega(s'-)<-2^L$.
  Therefore, $\uomega(s'')>2^L$ or $\lomega(s'')<-2^L$ for some $s''<s'\le s\le\sigma_L$,
  which contradicts the definition of $\sigma_L$.
  Let us now check that $\mathfrak{S}$
  (which we already know to be nonnegative in the original game)
  witnesses that $E$ is null.
  If $(\omega,\lomega,\uomega,\phi)\in E$, there is a constant $c$
  bounding $-\lomega|_{[0,t]}$ and $\uomega|_{[0,t]}$ from above.
  Any addend in~\eqref{eq:witness} for which $2^L>c$ will tend to infinity as $s\to\infty$
  (and in fact will be infinite at $s=t$).

  In the rest of this proof Reality is constrained to $\sup_s\left|\omega(s)\right|\le c$.
  Without loss of generality, set $c:=0.5$.
  We follow the same scheme as for Theorem~\ref{thm:continuous},
  again defining $x_k$ by \eqref{eq:x} and $x_{k,s}$ by \eqref{eq:x-ks},
  with the same $b_n$.
  Notice that, for $n\ge2$, we always have
  \begin{equation}\label{eq:phi}
    \left|
      \phi(a'_{k-1})
      -
      \phi(a''_{k-1})
    \right|
    \le
    2^{-n+1}
  \end{equation}
  in \eqref{eq:x} and \eqref{eq:x-ks};
  therefore, we can replace \eqref{eq:bound-1} by
  \begin{equation*}
    \left|x_k\right|
    \le
    b_n 2^{-n+1}
    \le
    0.5
  \end{equation*}
  (with the analogous inequalities for $x_{k,s}$),
  where the last inequality is true $n\ge8$,
  which we assume from now on in this proof.

  Essentially the same argument as in Section~\ref{sec:continuous} shows that \eqref{eq:square} still holds.
  Indeed, it suffices to check \eqref{eq:sup}.
  The nonnegativity of the process~$\tilde\K^n$ follows, for sufficiently large $n$,
  from $\left|\omega_{[0,t]}\right|\le0.5$;
  namely, when $n^4 2^{-2n+2}0.25\le n^{-3}$,
  $\tilde\K^n$ will be nonnegative even when the addend $\sum_{k=1}^{\infty}(\cdots)^2$ in~\eqref{eq:sum} is ignored.
  Applying Lemma~\ref{lem:O} now again gives \eqref{eq:square}.

  The proof is now completed in the same way as the proof of Theorem~\ref{thm:continuous}.
\end{proof}

\section{The case of a predictably non-c\`adl\`ag integrand and a continuous integrator}
\label{sec:non-cadlag}

In this section we will consider non-c\`adl\`ag integrands $\phi$,
motivated by, first of all, Tanaka's formulas,
which involve integrators such as $\phi(t):=\III_{\omega(t)>a}$ (lower semicontinuous for continuous $\omega$),
$\phi(t):=\III_{\omega(t)\ge a}$ (upper semicontinuous for continuous $\omega$),
or $\phi(t):=\sign(\omega(t)-a)$ (in general neither).
Such functions are not even regulated:
they have \emph{essential discontinuities}
(i.e., points $t$ such that $\phi(t-)$ or $\phi(t+)$ do not exist).
We will define the It\^o integral $\int\phi\dd\omega$ for such $\phi$
in the case where there is some kind of synergy between $\phi$ and $\omega$:
very roughly, we will require that $\omega$ does not change much around the essential discontinuities of $\phi$
(which will cover the examples given at the beginning of this paragraph).
The results of this section are very preliminary;
in particular,
in this version of the paper we only consider the case of continuous~$\omega$.

We will require that the integrand $\phi$ be non-c\`adl\`ag ``in a predictable manner''.
Formally, we now  define the sample space $\Omega$ to be the set
\begin{multline}\label{eq:sample-non-cadlag}
  \Omega
  :=
  \Bigl\{
    (\omega,\phi,\lphi,\uphi) \in C[0,\infty)\times\bbbr^{[0,\infty)}\times D[0,\infty)^2
    \st
    \forall t\in(0,\infty):\\
    \lphi(t)\le\liminf_{s\downarrow t}\phi(s)\le\limsup_{s\downarrow t}\phi(s)\le\uphi(t)
    \And \lphi(t)\le\phi(t)\le\uphi(t);\\
    \text{the set $\{t\in[0,\infty)\st\lphi(t)\ne\uphi(t)\}$ is closed}
  \Bigr\}
\end{multline}
consisting of quadruples 
\begin{equation}\label{eq:o}
  o = (\omega,\phi,\lphi,\uphi)
\end{equation}
such that, intuitively,
$\lphi$ and $\uphi$ control the non-c\`adl\`ag jumps of $\phi$.

The $\Omega$ of the previous section, \eqref{eq:sample-cadlag},
under the restriction that $\omega$ is continuous,
embeds into the $\Omega$ of this section, \eqref{eq:sample-non-cadlag},
by setting $\lphi:=\phi$ and $\uphi:=\phi$.
Analogously to that section, we set
\begin{equation*}
  o(t)
  :=
  (\omega(t),\phi(t),\lphi(t),\uphi(t)),
  \quad
  t\in[0,\infty),
\end{equation*}
and equip the sample space $\Omega$ with the universal completion $\FFF$
of the $\sigma$-algebra generated by the functions $o\in\Omega\mapsto o(t)$, $t\in[0,\infty)$,
so that the definitions of events, random variables,
stopping times $\tau$, $\tau$-measurable random variables,
simple trading strategies, etc.,
carry over to this case as well.

One of the conditions of the main result (Theorem~\ref{thm:non-cadlag}) of this section
will involve a slight modification of the standard notion of box-counting dimension
(analogous to the modification of Riemann integrals to Riemann--Stiltjes integrals).
For an interval $I$ of the real line and a real-valued function $f$ defined on $I$,
the \emph{oscillation} of $f$ over $I$ is
\begin{equation*}
  \osc_I(f)
  :=
  \sup_{t_1,t_2\in I}
  \left|
    f(t_1) - f(t_2)
  \right|
  =
  \sup_{t\in I}f(t)
  -
  \inf_{t\in I}f(t).
\end{equation*}
Let $\omega$ be a real-valued function defined on $[0,\infty)$ and $E$ be a bounded subset of $[0,\infty)$.
Set
\begin{equation*}
  M_{\omega}(E,\epsilon)
  :=
  \min
  \left\{
    k\ge1
    \st
    \exists I_1\cdots\exists I_k:
    E \subseteq \bigcup_{i=1}^k I_i
    \And
    \max_{i=1,\ldots,k}\osc_{I_i}(\omega)\le\epsilon
  \right\},
\end{equation*}
where $I_1,I_2,\ldots$ range over the intervals of $[0,\infty)$,
and set
\begin{align*}
  \dim_{\omega}(E)
  :=
  \limsup_{\epsilon\downarrow0}
  \frac{\log M_{\omega}(E,\epsilon)}{\log(1/\epsilon)}.
\end{align*}
For the identity function $\omega(t):=t$, $\forall t\in[0,\infty)$,
this becomes the usual definition of upper box-counting dimension
(also known as Minkowski dimension,
although it was first given in this form only by Pontryagin and Shnirel'man \cite{Pontryagin/Shnirelman:1932}).

Let us say that \eqref{eq:o} is \emph{tame} at time $t$,
$\tame_t(o)$ in symbols,
if $\dim_{\omega}(E)<2$, where $E:=\{s\in[0,t]\st\lphi(s)\ne\uphi(s)\}$.

Now the definition \eqref{eq:T-cadlag} of $T^n_k$ is modified by setting $T^n_0(o):=0$
and, for $k\ge1$:
\begin{itemize}
\item
  if $\lphi(T^n_{k-1})=\uphi(T^n_{k-1})$,
  \begin{multline*}
    T^n_k(o)
    :=
    \inf
    \Bigl\{
      t>T^n_{k-1}(o)
      \st
      \left|\omega(t)-\omega(T^n_{k-1})\right| \vee \left|\phi(t)-\phi(T^n_{k-1})\right| \ge 2^{-n}\\
      \text{ or }
      \lphi(t)<\uphi(t)
    \Bigr\};
  \end{multline*}
\item
  if $\lphi(T^n_{k-1})<\uphi(T^n_{k-1})$,
  \begin{equation*}
    T^n_k(o)
    :=
    \inf
    \Bigl\{
      t>T^n_{k-1}(o)
      \st
      \left|\omega(t)-\omega(T^n_{k-1})\right| \ge 2^{-n}
    \Bigr\}.
  \end{equation*}
\end{itemize}
(The requirement ``the set $\{t\st\lphi(t)\ne\uphi(t)\}$ is closed'' in~\eqref{eq:sample-cadlag}
is a simple way to ensure that $T^n_k$ are indeed stopping times.)
With this change, the definition of $(\phi\cdot\omega)^n_s$ is \eqref{eq:integral}.
The analogue of Lemma~\ref{lem:basic} is obvious.

\begin{theorem}\label{thm:non-cadlag}
  For each $t>0$,
  $(\phi\cdot\omega)^n_s$
  converges uniformly over $s\in[0,t]$ almost surely
  on the event $\tame_t$
  as $n\to\infty$.
\end{theorem}

Before we prove Theorem~\ref{thm:non-cadlag},
we briefly discuss its statement, especially the condition $\tame^{\omega}_t(\phi)$.
A more detailed statement of Theorem~\ref{thm:non-cadlag} would be:
for each $t>0$,
the set of \eqref{eq:o} such that
\begin{itemize}
\item
  the sequence of functions $s\in[0,t]\mapsto(\phi\cdot\omega)^n_s$ fails to converge
  in the uniform metric on $[0,t]$
  as $n\to\infty$
\item
  and $\tame_t(o)$
\end{itemize}
is null.

The following lemma shows that the condition $\tame_t$ is mild in a certain sense.

\begin{lemma}\label{lem:mild}
  Almost surely,
  for any $t>0$, $\dim_{\omega}([0,t])=2$.
\end{lemma}

\begin{proof}
  Let us use Theorem~3.1 in \cite{\CTIV} (a probability-free version of the Dubins--Schwarz result).
  The quadratic variation of $\omega$ was defined in Section~\ref{sec:variation},
  but here we can also use the definition given in \cite{\CTIV} (and not involving $\phi$).
  Since the union of countably many null events is null,
  it suffices to prove that $\dim_{\omega}([0,\tau])=2$,
  where $\tau$ is the first time when the quadratic variation of $\omega$ reaches a given constant $c>0$.
  This reduces our task to proving that $\dim_{\omega}([0,c])=2$ for a typical path $\omega$ of standard Brownian motion.
  If we divide $[0,c]$ into $n$ equal parts of length $c/n$, L\'evy's modulus of continuity theorem
  (see, e.g., \cite{Morters/Peres:2010}, Theorem~1.14) shows
  that the oscillation of $\omega$ on each part is equivalent to $\sqrt{2(c/n)\ln n}$ or less as $n\to\infty$.
  For $\epsilon:=\sqrt{2(c/n)\ln n}$ we get
  $$
    M_{\omega}([0,c],\epsilon)\le n
    =
    O
    \left(
      \frac{1}{\epsilon^2}
      \log\frac{1}{\epsilon}
    \right),
  $$
  which remains true as $\epsilon\downarrow0$.
\end{proof}

Lemma~\ref{lem:mild} can be interpreted to say that the condition $\dim_{\omega}(E)<2$
implicit in Theorem~\ref{thm:non-cadlag} means that $E$ is only slightly less massive than the whole of $[0,t]$.
On the other hand, the next lemma shows that the sets $\{\omega=a\}$ of essential discontinuities
in Tanaka's formulas are typically much less massive.

\begin{lemma}
  Let $a\in\bbbr$.
  Almost surely, for any $t>0$,
  \begin{equation}\label{eq:dim}
    \dim_{\omega}(\{\omega=a\}\cap[0,t])
    \le
    1.
  \end{equation}
\end{lemma}

\begin{proof}
  Similarly to the proof of Lemma~\ref{lem:mild},
  it suffices to prove~\eqref{eq:dim}
  for a typical path $\omega$ of standard Brownian motion.
  And in this case \eqref{eq:dim} follows immediately from, e.g.,
  the downcrossing representation of the local time at zero (or at $a$):
  see, e.g., \cite{Morters/Peres:2010}, Theorem~6.1.
\end{proof}

\begin{proof}[Proof of Theorem~\ref{thm:non-cadlag}]
  Fix $t>0$.
  As in the proof of Theorem~\ref{thm:cadlag},
  we will assume, without loss of generality, that $\left|\omega\right|\le0.5$;
  we will also assume, again without loss of generality,
  that the analogous inequalities hold for $\phi$, $\lphi$, and $\uphi$.
  (This will ensure $\left|x_{k,s}\right|\le b_n2^{-n}\le0.5$, assuming $n\ge7$.)
  Finally, since the union of countably many null sets is null, there is no loss of generality
  in replacing the condition $\dim_{\omega}(E)<2$ by $\dim_{\omega}(E)<2-\delta$ for a given $\delta>0$,
  where $E:=\{s\in[0,t]\st\lphi(t)<\uphi(t)\}$.
  Therefore, we assume that, as $\epsilon\downarrow0$, $E$ can be covered by $O(\epsilon^{\delta-2})$ intervals $I$
  such that $\osc_I(\omega)\le\epsilon$.
  This implies that the number of $k$ such that $T^n_k\le t$ and $\lphi(T^n_k)<\uphi(T^n_k)$ is
  \begin{equation}\label{eq:O-2}
    O(2^{(2-\delta)n}).
  \end{equation}

  It suffices to show that \eqref{eq:square} still holds.
  Without loss of generality we assume that the sum is over the $k$ such that $a_{k-1}<t$.
  We divide such $k$ into four kinds
  (and in each of the four corresponding items below,
  the default is that $\sum_k$ stands for the sum over the $k$ of the kind considered in that item):
  \begin{enumerate}
  \item
    The first kind of $k$ are those satisfying $\lphi(a_{k-1})<\uphi(a_{k-1})$
    (remember that we are only interested in $k$ such that $a_{k-1}<t$).
    Since we always have $\left|\omega(a_k\wedge s)-\omega(a_{k-1}\wedge s)\right|\le2^{-n}$,
    \begin{equation*}
      \sup_{s\in[0,t]}
      \sum_k x_{k,s}^2
      =
      O
      \left(
        2^{(2-\delta)n}
        b_n^2 2^{-2n}
      \right)
      =
      o(1).
    \end{equation*}
    (We have used the fact that both the number of $k$ such that $T^n_k\le t$ and $\lphi(T^n_k)<\uphi(T^n_k)$ is \eqref{eq:O-2}
    and the number of $k$ such that $T^{n-1}_k\le t$ and $\lphi(T^{n-1}_k)<\uphi(T^{n-1}_k)$ is \eqref{eq:O-2}.)
  \item
    The second kind of $k$ are those satisfying
    $$
      \lphi(a_{k-1})=\uphi(a_{k-1}) \And \lphi(a''_{k-1})<\uphi(a''_{k-1}).
    $$
    Such $k$ satisfy $a''_{k-1}<a_{k-1}$ and $a'_{k-1}=a_{k-1}$.
    We cannot claim that the number of such $k$ is \eqref{eq:O-2}
    since different $k$ may lead to the same value of $a''_{k-1}$,
    so we will need a more complicated argument making use of the K29 martingale \eqref{eq:K29}.
    First we notice that
    \begin{multline}\label{eq:martingale}
      n^{-10}
      +
      \sum_k
      \left(
        \omega(a_k\wedge s)-\omega(a_{k-1}\wedge s)
      \right)^2\\
      -
      \sum_{\substack{k=1,2,\ldots:T^{n-1}_{k-1}<t,\\\lphi(T^{n-1}_{k-1})<\uphi(T^{n-1}_{k-1})}}
      \left(
        \omega(T^{n-1}_k\wedge s)-\omega(T^{n-1}_{k-1}\wedge s)
      \right)^2
    \end{multline}
    is the value of a simple capital process at time $s\le t$.
    The subtrahend in~\eqref{eq:martingale} is
    \begin{equation*}
      O
      \left(
        2^{(2-\delta)n}
        2^{-2n}
      \right)
      =
      o(n^{-10}),
    \end{equation*}
    and so the martingale value~\eqref{eq:martingale} is nonnegative,
    even if we ignore its second addend $\sum_k(\cdots)^2$.
    Let us make the simple capital process nonnegative by stopping trading when it becomes zero;
    this will not affect the process at all for large enough $n$.
    By Lemma~\ref{lem:O}, \eqref{eq:martingale} is $O(n^{-8})$ uniformly in $s\in[0,t]$,
    and so the second addend of~\eqref{eq:martingale} is $O(n^{-8})$, a.s.
    Therefore,
    \begin{multline*}
      \sup_{s\in[0,t]}
      \sum_k x_{k,s}^2
      =
      \sup_{s\in[0,t]}
      \sum_k
      b_n^2
      \left(
        \phi(a_{k-1})
        -
        \phi(a''_{k-1})
      \right)^2\\
      \times
      \left(
        \omega(a_k)\wedge s
        -
        \omega(a_{k-1})\wedge s
      \right)^2
      =
      O
      \left(
        n^4
        n^{-8}
      \right)
      =
      o(1)
      \quad
      \text{a.s.}
    \end{multline*}
  \item
    The third kind of $k$ are those satisfying
    $$
      \lphi(a_{k-1})=\uphi(a_{k-1}) \And \lphi(a'_{k-1})<\uphi(a'_{k-1}).
    $$
    Such $k$ are treated in the same way as the $k$ of the second kind.
  \item
    The last kind of $k$ are those for which
    $$
      \lphi(a_{k-1})=\uphi(a_{k-1}) \And \lphi(a'_{k-1})=\uphi(a'_{k-1}) \And \lphi(a''_{k-1})=\uphi(a''_{k-1}).
    $$
    Such $k$ satisfy \eqref{eq:phi},
    and we again have
    \begin{equation*}
      \sup_{s\in[0,t]}
      \sum_{k} x_{k,s}^2
      =
      o(1)
      \quad
      \text{a.s.}
      \qedhere
    \end{equation*}
  \end{enumerate}
\end{proof}

\section{Conclusion}
\label{sec:conclusion}

The most obvious directions of further research are:
\begin{itemize}
\item
  to explore the dependence of $\int\phi\dd\omega$ on the choice of the partitions $T^n_k$,
\item
  to extend Theorem~\ref{thm:Ito} to convex functions $F$,
\item
  and to relax the conditions of Theorem~\ref{thm:non-cadlag}.
\end{itemize}

\subsection*{Acknowledgments}

Thanks to Rafa\l\ \L{}ochowski for useful discussions and informing me about \cite{Lochowski:2015},
to Glenn Shafer for his comments,
to Nicolas Perkowski and David Pr\"omel for clarifying the relations with \cite{\PerkowskiPromelIntegrals},
and to Nicolas for a useful discussion of the desiderata for stochastic integration.
This research was supported by the Air Force Office of Scientific Research
(grant FA9550-14-1-0043).

\appendix
\renewcommand{\thesection}{A}
\refstepcounter{section}
\label{sec:A}
\section*{Appendix A: Useful discrete-time supermartingales}
\addcontentsline{toc}{section}{Appendix A: Useful discrete-time supermartingales}

Our proofs of Theorems~\ref{thm:continuous}, \ref{thm:cadlag}, and~\ref{thm:non-cadlag}
are based on a simple large-deviation-type supermartingale,
which will be defined in this appendix,
and on a classical martingale going back to \cite{Kolmogorov:1929LLN},
to be defined in Appendix~B below.

We consider the case of discrete time,
namely, the following perfect-information protocol:

\bigskip

\noindent
\textsc{Betting on bounded below variables}

\smallskip

\noindent
\textbf{Players:} Sceptic and Reality

\smallskip

\noindent
\textbf{Protocol:}

\parshape=5
\IndentI   \WidthI
\IndentI   \WidthI
\IndentII  \WidthII
\IndentII  \WidthII
\IndentII  \WidthII
\noindent
Sceptic announces $\K_0\in\bbbr$.\\
FOR $k=1,2,\dots$:\\
  Sceptic announces $M_k\in\bbbr$.\\
  Reality announces $x_k\in[-0.5,\infty)$.\\
  Sceptic announces $\K_k \le \K_{k-1} + M_k x_k$.

\bigskip

\noindent
We interpret $\K_k$ as Sceptic's capital at the end of round $k$.
Notice that Sceptic is allowed to choose his initial capital $\K_0$
and to throw away part of his money at the end of each round.

A \emph{process} is a real-valued function defined on all finite sequences
$(x_1,\ldots,x_K)$,
$K=0,1,\ldots$,
of Reality's moves.
If we fix a strategy for Sceptic,
his capital $\K_K$, $K=0,1,\ldots$, will become a process.
Such processes are called \emph{supermartingales}.
\begin{lemma}\label{lem:super}
  The process
  \begin{equation}\label{eq:super-1}
    \K_K
    :=
    \prod_{k=1}^K
    \exp
    \left(
      x_k
      -
      x_k^2
    \right)
  \end{equation}
  is a supermartingale.
\end{lemma}

\noindent
We do not require the measurability of supermartingales \emph{a priori},
but \eqref{eq:super-1} is, of course, measurable.
The corresponding strategy for Sceptic used in the proof will be $M_k:=\K_{k-1}$,
and so will also be measurable.
The lemma will still be true if the interval $[-0.5,\infty)$ in the protocol is replaced by $[-0.683,\infty)$
(but will no longer be true for $[-0.684,\infty)$).

\begin{proof}
  It suffices to prove that on round $k$ Sceptic can turn a capital of $\K>0$
  into a capital of at least
  \begin{equation*}
    \K
    \exp
    \left(
      x_k
      -
      x_k^2
    \right);
  \end{equation*}
  in other words,
  that he can obtain a payoff $M_k x_k$ of at least
  \begin{equation*}
    \exp
    \left(
      x_k
      -
      x_k^2
    \right)
    -
    1.
  \end{equation*}
  This will follow from the inequality
  \begin{equation*}
    \exp
    \left(
      x_k
      -
      x_k^2
    \right)
    -
    1
    \le
    x_k.
  \end{equation*}
  Setting $x:=x_k$, moving the $1$ to the right-hand side, and taking logs of both sides,
  we rewrite this inequality as
  \begin{equation*}
    x - x^2
    \le
    \ln(1+x),
  \end{equation*}
  where $x\in[-0.5,\infty)$.
  Since we have an equality for $x=0$,
  it remains to notice that the derivative of the left-hand side of the last inequality
  never exceeds the derivative of its right-hand side for $x>0$,
  and that the opposite relation holds for $x<0$.
\end{proof}

Another useful process is
\begin{equation}\label{eq:super-2}
  \frac12
  \left(
    \prod_{k=1}^K
    \exp
    \left(
      x_k
      -
      x_k^2
    \right)
    +
    \prod_{k=1}^K
    \exp
    \left(
      -x_k
      -
      x_k^2
    \right)
  \right),
\end{equation}
which is a supermartingale in the protocol of betting on bounded variables,
where Reality is required to announce $x_k\in[-0.5,0.5]$.
(It suffices to apply Lemma~\ref{lem:super} to $x_k$ and $-x_k$
and to average the resulting supermartingales.)

\begin{remark}
  In this appendix we used the method described in \cite{Sasai/etal:2015}, Section~2;
  in fact, it is shown (using slightly different terminology) in \cite{Sasai/etal:2015} that
  \begin{equation*}
    \prod_{k=1}^K
    \exp
    \left(
      x_k
      -
      \frac{x_k^2}{2}
      -
      \left|x_k\right|^3
    \right)
  \end{equation*}
  is a supermartingale in the protocol of betting on bounded variables, $\left|x_k\right|\le\delta$
  for a small enough $\delta>0$
  (it is sufficient to assume $\delta\le0.8$).
\end{remark}

\renewcommand{\thesection}{B}
\refstepcounter{section}
\label{sec:B}
\section*{Appendix B: Another useful discrete-time supermartingale}
\addcontentsline{toc}{section}{Appendix B: A useful discrete-time martingale}

In this appendix we will define another process
used in the proofs of the main results of this paper
(in principle, we could have also used this process to replace in those proofs
the process defined in Appendix~A).

We still consider the case of discrete time.
The perfect-information protocol of this section is:

\bigskip

\noindent
\textsc{Betting on arbitrary variables}

\smallskip

\noindent
\textbf{Players:} Sceptic and Reality

\smallskip

\noindent
\textbf{Protocol:}

\parshape=5
\IndentI   \WidthI
\IndentI   \WidthI
\IndentII  \WidthII
\IndentII  \WidthII
\IndentII  \WidthII
\noindent
Sceptic announces $\K_0\in\bbbr$.\\
FOR $k=1,2,\dots$:\\
  Sceptic announces $M_k\in\bbbr$.\\
  Reality announces $x_k\in\bbbr$.\\
  $\K_k := \K_{k-1} + M_k x_k$.

\bigskip

\noindent
Sceptic's capital $\K_K$ as function of Reality's moves $x_1,\ldots,x_K$
for a given strategy for Sceptic
is a process called a \emph{martingale}
(this term is natural as our new protocol does not allow Sceptic to throw money away).

\begin{lemma}
  The process
  \begin{equation}\label{eq:K29}
    \K_K
    :=
    \sum_{k=1}^K
    x_k^2
    -
    \left(
      \sum_{k=1}^K
      x_k
    \right)^2
  \end{equation}
  is a martingale.
\end{lemma}

We will refer to \eqref{eq:K29} as the \emph{K29 martingale}.

\begin{proof}
  The increment of \eqref{eq:K29} on round $K$ is
  \begin{equation}\label{eq:increment}
    x_K^2
    -
    \left(
      \sum_{k=1}^K
      x_k
    \right)^2
    +
    \left(
      \sum_{k=1}^{K-1}
      x_k
    \right)^2
    =
    -2
    \left(
      \sum_{k=1}^{K-1}
      x_k
    \right)
    x_K
  \end{equation}
  and, therefore, is indeed of the form $M_K x_K$.
\end{proof}
\end{document}